\documentclass{article}
\usepackage{fullpage}
\usepackage{amsmath,amsthm,amssymb,graphics}
\usepackage{algorithm,algorithmicx,algpseudocode}
\newtheorem{theorem}{Theorem}
\providecommand{\keywords}[1]{\textbf{\textit{Keywords---}} #1}

\begin{document}

\title{Short Combinatorial Proof that the DFJ Polytope is contained in the MTZ Polytope for the Asymmetric Traveling Salesman Problem}

\author{Mark Velednitsky \thanks{Fellow, National Physical Science Consortium} \\ UC Berkeley \\ marvel@berkeley.edu}

\date{}

\maketitle

\begin{abstract}
For the Asymmetric Traveling Salesman Problem (ATSP), it is known that the Dantzig-Fulkerson-Johnson (DFJ) polytope is contained in the Miller-Tucker-Zemlin (MTZ) polytope. The analytic proofs of this fact are quite long. Here, we present a proof which is combinatorial and significantly shorter by relating the formulation to distances in a modified graph.
\end{abstract}

\section{Introduction}
The Asymmetric Traveling Salesman Problem (ATSP) on the graph $G = (V, A)$ is typically formulated as an Integer Program (IP) by assigning each arc $(i,j)$, of weight $c_{ij}$, a binary variable $x_{ij}$ indicating whether or not it participates in the tour:
\begin{align}
	\text{minimize }  &\sum_{(i,j)\in A}{c_{ij}x_{ij}} & \nonumber \\
	\text{subject to }  &\sum_j{x_{ij}} = 1 \ \ &\forall i \in V \nonumber \\
	&\sum_i{x_{ij}} = 1 \ \  &\forall j \in V \nonumber \\
	&\text{no sub-tours in} \ \{(i,j)|x_{ij}=1\} \ \ & \label{eqline:subtour} \\
	&x_{ij} \in \{0,1\} \ \ &\forall (i,j) \in A \nonumber \\
	\nonumber
\end{align}
Several variants of the sub-tour elimination constraint (\ref{eqline:subtour}) have been proposed. The DFJ constraints are:
\begin{equation}
\label{eq:DFJ}
\sum_{i \in Q}{\sum_{j \in Q}{x_{ij}}} \leq |Q|-1
\end{equation}
for any $Q \subseteq \{2,3,\ldots,n\}$.
The MTZ constraints introduce a new variable $u_i$ at each node $i \in V$ such that \cite{miller1960integer}:
\begin{equation}
\label{eq:MTZ}
u_i - u_j + n x_{ij} \leq n-1 \ \forall (i,j) \in A
\end{equation}

The $u_i$ are meant to enumerate the order in which nodes appear in the tour. That is, $u_i = 1$ for the first node, $u_i = 2$ for the second, and so on.

The DFJ and MTZ polytopes are the feasible regions of the respective LP relaxations. It is known that the MTZ formulation produces a weaker LP relaxation. However, rigorous proofs of this fact are quite involved \cite{desrochers1991improvements,gouveia1999asymmetric,langevin1990classification,padberg1991analytical}.

Even though they are weaker, MTZ-like constraints have been applied to Vehicle Routing Problems and are popular for solving small instances of ATSP \cite{bektacs2014requiem}. Having a concise proof of their weakness could be instructive for understanding the constraints, teaching them, and applying them elsewhere \cite{pataki2003teaching}.

\section{Proof}
\begin{theorem}
The DFJ polytope is contained in the MTZ polytope.
\end{theorem}

\begin{proof}
Let $x_{ij}$ be feasible for formulation DFJ. We define a new graph $G$ where the arc weights are $(n-1) - n x_{ij}.$ We let $-u_j$ be the length of the shortest path from $1$ to $j$ in $G$. We claim that these $u_j$ are well-defined and make the $u_j$ and $x_{ij}$ together satisfy formulation MTZ. To check that MTZ is satisfied, we write the shortest path condition in $G$: 
$$-u_j \leq -u_i + (n-1) - n x_{ij} \implies u_i - u_j + n x_{ij} \leq (n-1).$$ 
To confirm that the $u_j$ are well-defined, we need to prove there are no negative-cost cycles in $G$. Assume there is a negative cost cycle with edge set $C$ with node set $Q$:
$$\sum_{(i,j) \in C}{((n-1) - n x_{ij})} < 0 \implies |Q| (n-1) - n \sum_{(i,j) \in C}{x_{ij}} < 0 \implies |Q| \frac{n-1}{n} < \sum_{(i,j) \in C}{x_{ij}}.$$
But the conditions of formulation DFJ give us
$$\sum_{(i,j) \in C}{x_{ij}} \leq |Q| - 1.$$
This is a contradiction (since $|Q| = |C| \leq n$), so there are no negative cost cycles.
\end{proof}

\bibliography{salesman_polytopes_ref}
\bibliographystyle{plain}

\end{document}